%% file: Localization_NLOS_SDP_arxiv.tex
\documentclass[conference]{IEEEtran}             
\IEEEoverridecommandlockouts
\include{packages_all}

\usepackage{algorithm}
\usepackage{algorithmic}
\begin{document}

\title{Semi-Definite Programming Relaxation for Non-Line-of-Sight Localization \thanks{This project is supported in part by AFOSR grant FA9550-10-1-0567}}

\author{ \IEEEauthorblockN{Venkatesan. N. Ekambaram, Giulia Fanti and Kannan Ramchandran}\\
\IEEEauthorblockA{Department of EECS, University of California, Berkeley\\
Email: \{venkyne, gfanti, kannanr\}@eecs.berkeley.edu}}

\maketitle

\begin{abstract}
We consider the problem of estimating the locations of a set of points in a $k$-dimensional euclidean space given a subset of the pairwise distance measurements between the points. We focus on the case when some fraction of these measurements can be arbitrarily corrupted by large additive noise. Given that the problem is highly non-convex, we propose a simple semidefinite programming relaxation that can be efficiently solved using standard algorithms. We define a notion of non-contractibility and show that the relaxation gives the exact point locations when the underlying graph is non-contractible. The performance of the algorithm is evaluated on an experimental data set obtained from a network of 44 nodes in an indoor environment and is shown to be robust to non-line-of-sight errors.
\end{abstract}
{\bf \em Keywords:} Non-Line-of-Sight localization, semi-definite programming, robust matrix decomposition.

\section{Introduction}

The problem of localization has applications in many interesting areas such as cyber physical systems \cite{lee2008cyber}, molecular biology \cite{prot}, nonlinear dimensionality reduction \cite{roweis2000nonlinear} etc. In all these applications, pairwise noisy distance measurements are obtained between subsets of points/nodes in some Euclidean space and it is required to estimate the locations of these points with high accuracies. Given the range of applications, the problem has received considerable attention in the last decade and is an active area of research. Even though the problem statement looks deceptively simple,  the general case is shown to be NP-complete \cite{aspnes2004computational}, i.e. finding a valid configuration of points satisfying a subset of pairwise distance measurements is not solvable in polynomial time. Significant research work is devoted to obtaining convex relaxations for the underlying optimization problem and obtaining conditions under which the problem can be solved in polynomial time. Semi-definite programming (SDP) relaxations have been proposed \cite{biswas2004semidefinite}, \cite{zhu2010universal}, \cite{javanmard2011localization} and the relaxation has been shown to obtain the exact solution for certain classes of graphs known as uniquely localizable graphs, examples of which include  d-lateration and random geometric graphs.\\

 \begin{figure}[h!]
 \centering
\subfigure[\scriptsize Sensor Localization.]
{\includegraphics[height =1.3in]{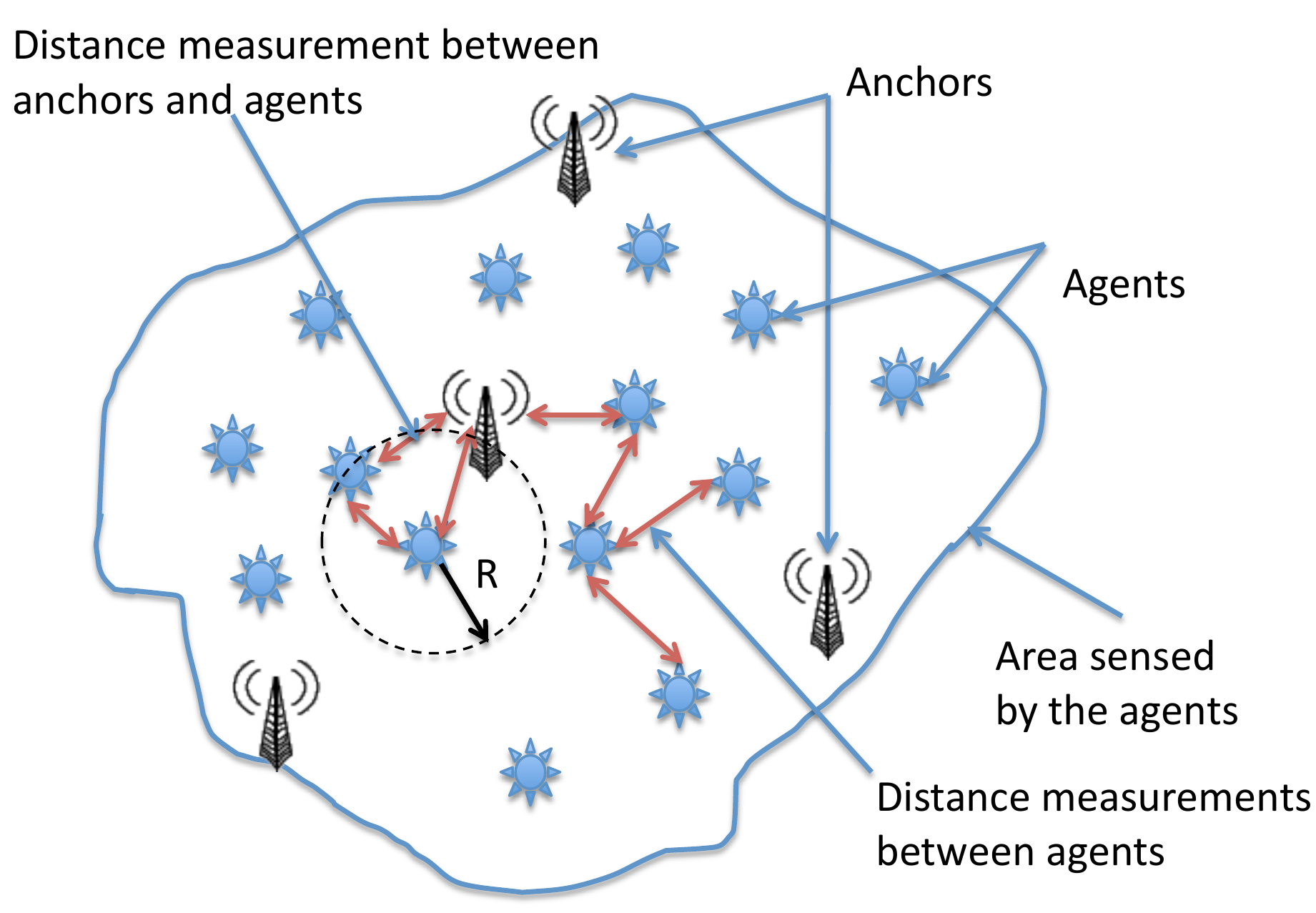} \label{fig:SensorLoc}}
\subfigure[\scriptsize Molecular conformation.]
{\includegraphics[height =1.2in]{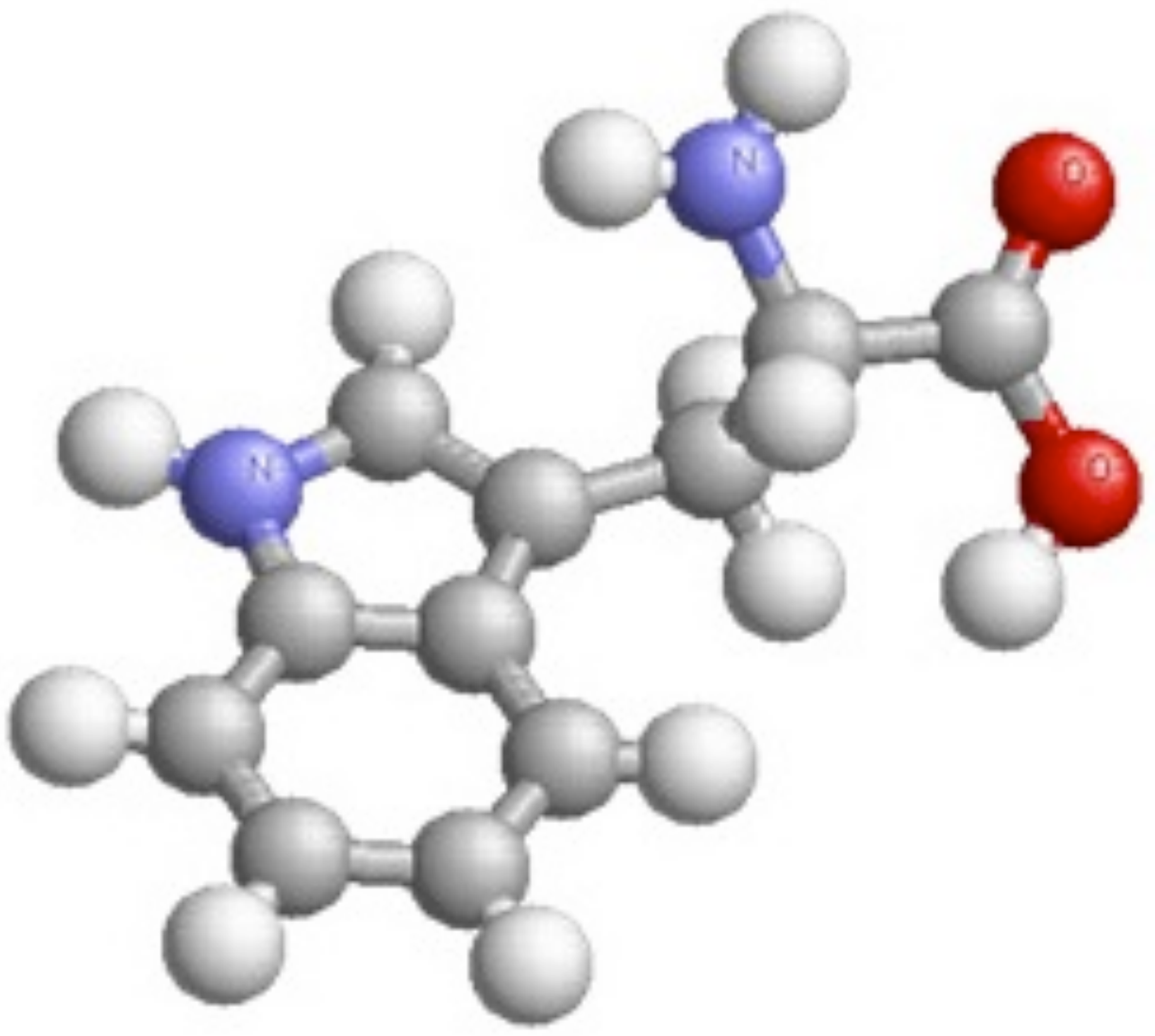} \label{fig:MolConf}}\\
\subfigure[\scriptsize Manifold Learning.]
{\includegraphics[height =2in]{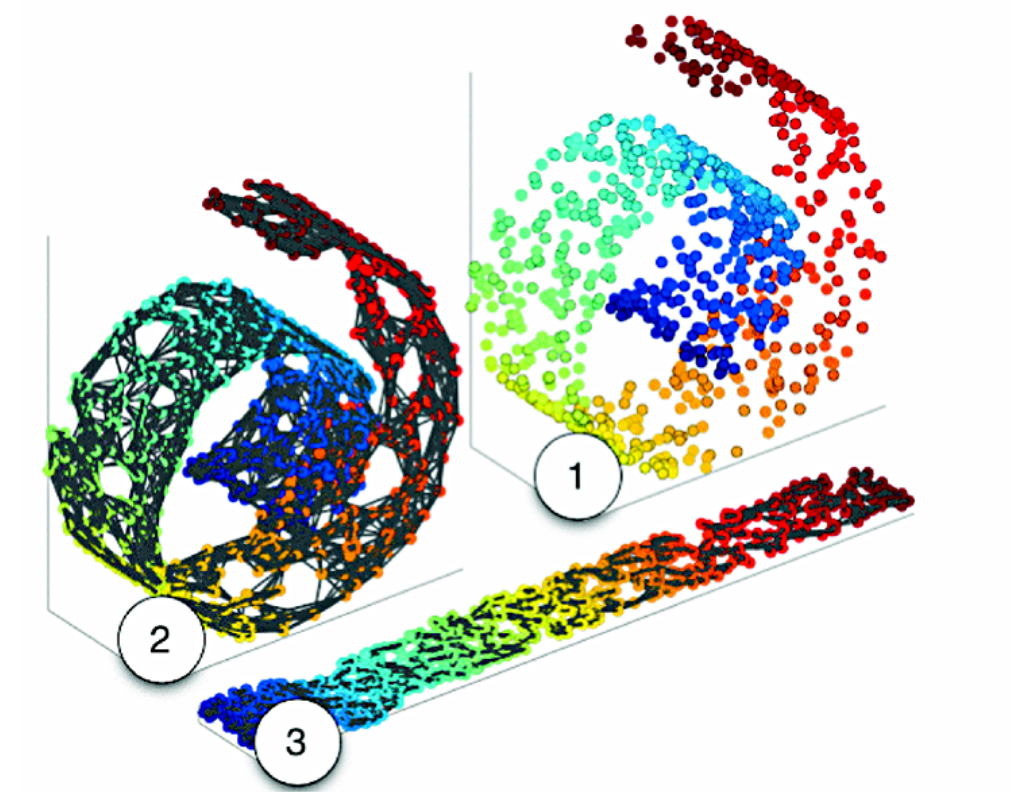}\label{fig:SwissRoll}}
\caption{(a) Sensors placed in a field for monitoring the region. (b) Tryptophan, one of the 20 standard amino acids \cite{prot}. (c) (1) A high dimensional ``swiss roll'' data set that spans a lower dimensional hypersurface \cite{dattorro2005convex}. (2) Data points within a local radius are connected. (3) The lower dimensional representation of the high dimensional  surface preserving the relative geometry of the points.}
\end{figure}

 Most of the existing literature on convex relaxations focus on the case where the pairwise proximity measurements are either known noiselessly or known with some slight perturbations (i.e. line-of-sight (LOS) localization). However in practice, a significant fraction of the data could have outliers. For example, in the case of sensor localization,  non-line-of-sight (NLOS) propagation introduces large errors in the distance measurements. Spin-diffusion phenomenon in NMR spectroscopy for protein molecular conformation is known to produce bias errors in the measurements \cite{berger1999reconstructing}. Data in higher dimensions do not all lie on the manifold and there could be outliers that lie far away from the surface of the manifold.\\
 
   Traditional methods that do not account for the large bias errors can result in significant localization errors. The hardness of the problem arises from the fact that the bias errors are typically large and it is not known a priori which measurements are biased. There has been some work in the literature that addresses this case using convex relaxations \cite{venkatesh2006linear}. However these relaxations assume that the locations of the errors are known which makes the problem considerably simpler to solve. Our focus here is to address the problem of localization in the presence of large unknown bias errors.  For simplicity, we will keep the sensor localization problem in mind as we progress through the rest of the paper, though all the techniques can be easily adapted to other applications. We propose a semi-definite programming (SDP) relaxation for the case of NLOS localization and provide conditions under which the algorithm retrieves the node locations. The performance of the algorithm is evaluated through simulations and also validated on a data set collected through a real world indoor experimental setup of 44 nodes that is publicly available \cite{span}. The proposed relaxation is shown to mitigate the effect of the NLOS errors to a large extent by obtaining accurate estimates of the node locations. 

\section{Problem Setup}
Consider a static random placement of $N$ nodes in a $k$-dimensional euclidean space. The locations of these nodes are unknown and these nodes would hence be referred to as {\em agents}. We also have $M$ special nodes whose locations are exactly known called as {\em anchors}. Nodes within a radius $r$ of each other obtain relative distance measurements that could be corrupted by noise. Based on these distance measurements, the problem is that of determining the locations of the agents. Each of these measurements is modeled as  either a LOS-dominated signal or a NLOS-dominated signal by choosing the observation noise to be drawn from a mixture of two distributions.\\

Let $x_i \in {\mathbb R}^k$ denote the location of the $i$th agent. Let $d_{ij} = ||x_i-x_j||$ denote the actual distance between the $i$th and $j$th node. $\tilde{d}_{ij}$ denotes the corresponding distance measurement. We will model the distance measurement as $ \tilde{d}_{ij}^2  =  d_{ij}^2 + b_{ij} + n_{ij}$,
where $n_{ij}$ is a zero mean additive noise (e.g. gaussian) and $b_{ij}$ is given by,
\bean
b_{ij}  & = &  0  \mbox{   if  LOS}  \\ 
            & \geq & 0 \mbox{   otherwise}.
\eean
 Under NLOS, $b_{ij}$ is usually taken to be a positive random variable to model the statistics of multipath noise. 
The noise model is motivated by the time-of-arrival measurement modality for distance measurements. Since the primary multipath signal travels a longer distance than the true distance, the time delay of arrival for the NLOS path is larger than the LOS path and hence the NLOS bias can be modeled as additive. Similarly, the gaussian noise component is principally due to thermal noise in the receiver and so we take that term to be additive as well. We will assume that $\alpha$ fraction of the measurements are NLOS. We discuss a semi-definite programming approach to tackle this problem in the sections to follow.

\section{SDP Relaxation} 
\label{sec:SDPrelax}
Given the problem statement, we can define  a network $(X,A,{\cal E}_X,{\cal E}_A,D)$, of $N$ agents and $M$ anchors, where $X =  [x_1^T; x_2^T; ... ;x_N^T] \in \mathbb{R}^{N \times k}$ is the matrix of agent locations, $A = [a_1^T;...;a_M^T] \in \mathbb{R}^{M \times k}$ is the matrix of anchor locations, ${\cal E}_X$ is the set of edges between agents,  ${\cal E}_A$ is the set of edges between agents and anchors and $D$ is the matrix of pairwise distances between nodes that belong to the edge set. An edge is present between two nodes in the network iff we have a distance measurement between them (e.g. if they are within a distance $r$ of each other). We formulate the problem of estimating the node locations as an optimization problem. For this purpose, we will treat the biases $b_{ij}$'s in the measurements, as parameters that need to be estimated. The maximum likelihood formulation for estimating the node locations as well as the biases gives us the following optimization problem,
 \begin{equation}
\begin{aligned}
&\underset{X,B}{\text{maximize}}
&  p\left(\{\tilde{d}_{ij}^2\} | \{x_i\}_{i = 1}^N, \{a_i\}_{i=1}^M, \{b_{ij}\} \right) \\
& \text{subject to}
& b_{ij} \geq 0, \; i,j \in {\mathcal E_X \cup \mathcal E_A}.
\end{aligned}
\end{equation}

The above optimization reduces to the following least squares minimization if we take $n_{ij}$'s to be i.i.d gaussian.

 \begin{equation}
\begin{aligned}
& \underset{X, B}{\text{minimize}}
& \sum_{i,j \in {\mathcal E}_X} \left( \underbrace{||x_i-x_j||_2^2}_{d_{ij}^2} - \tilde{d}_{ij}^2+b_{ij}\right)^2  + \\
& & \sum_{i,j \in {\mathcal E}_A} \left( \underbrace{||x_i-a_j||_2^2}_{d_{ij}^2} - \tilde{d}_{ij}^2+b_{ij}\right)^2 \\
& \text{subject to}
& b_{ij} \geq 0, \; i,j \in {\mathcal E_X \cup \mathcal E_A}
\end{aligned}
\label{eq:min}
\end{equation}
where $B = \{b_{ij}\}$ is the matrix of bias variables. The above minimization problem is highly non-convex and has multiple minima. We aim at obtaining a convex approximation to the above problem.\\

 \begin{figure}
\centering
\includegraphics[height = 2.5in]{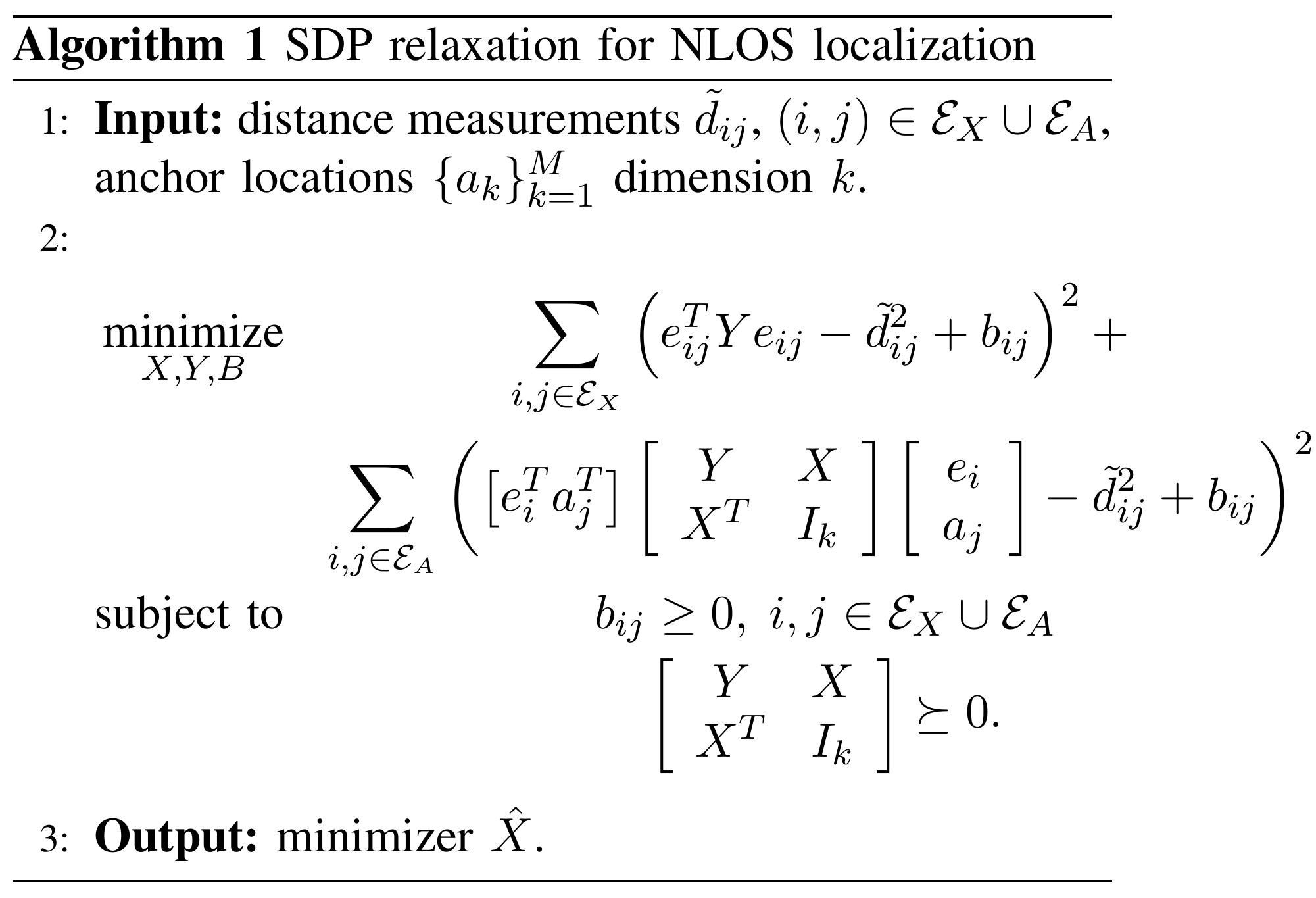}
\vspace{-0.5in}
\end{figure}

 The distance between the agents can be expressed as $||x_i-x_j||_2^2=e_{ij}^TXX^Te_{ij}$ where $e_{ij}$ is a vector in ${\mathbb{R}^N}$ with a $1$ and $-1$ at the $i$th and $j$th locations respectively and zeros elsewhere. Using the substitution $Y=XX^T$, we get $||x_i-x_j||_2^2=e_{ij}^TYe_{ij}$. The minimization problem (\ref{eq:min}) can thus be written as, 
\begin{equation}
 \begin{aligned}
\underset{X, Y, B}{\text{minimize}}
& \displaystyle \sum_{i,j \in {\mathcal E}_X} \left( e_{ij}^TYe_{ij} - \tilde{d}_{ij}^2+b_{ij}\right)^2  + \\
&  \displaystyle \sum_{i,j \in {\mathcal E}_A} \left( e_i^T Y e_i - 2 e_i^TXa_j + a_j^Ta_j - \tilde{d}_{ij}^2+b_{ij}\right)^2 \\
 \text{subject to}\\
& b_{ij} \geq 0, \; i,j \in {\mathcal E_X \cup \mathcal E_A} \\
& Y = XX^T,
\end{aligned}
 \label{eq:min2}
 \end{equation}
where $e_i$ is a vector with a one at the $i$th position and zeros elsewhere.\\

Since the constraint $Y=XX^T$ is not convex, we relax it to $Y\succeq XX^T$. Using Schur complements, we obtain the minimization problem as shown in Algorithm 1. The above semi-definite program can be efficiently solved using standard convex optimization techniques. One of the main advantages of this algorithm is that we do not need to input noise parameters to evaluate the optimization problem, which makes the algorithm useful in a practical setting.\\
 
 Let us develop an intuition for the relaxation. Without the relaxation, in problem (\ref{eq:min2}), we were searching for a feasible set of points in the $k$-dimensional space that minimize the cost function. The constraint that the points need to be in the $k$-dimensional space essentially makes the problem non-convex and hard to solve. The condition $Y\succeq XX^T$ relaxes the search for the points in any $k'$-dimensional space, such that $k' \leq N$. This can be seen as follows. If $k'$ is the rank of $Y$, then $Y\succeq XX^T$ implies that there exists $X' \in \mathbb{R}^{N \times (k'-k)}$ such that $Y = \left[X X' \right]\left[ {\begin{array}{cc}
 X^T \\
 X'^T \\
 \end{array} } \right] $.  $\tilde{X} = \left[X X' \right] \in \mathbb{R}^{N \times k'}$ can be thought of as the set of points in $k'$ dimensional space that are the minimizers of the optimization problem in Algorithm 1.  Figure \ref{fig:loc2d3d} illustrates this relaxation. The price that we pay for this simplification is that, once we get a solution in some $k'$ dimensional space, the projection onto the original $k$-dimensional space could yield bad results.The rest of the paper explores conditions under which this relaxation actually gives a solution in the $k$-dimensional space that would be close to the true node locations. 
 
  \begin{figure}
 \centering
 \includegraphics[height = 1.8in]{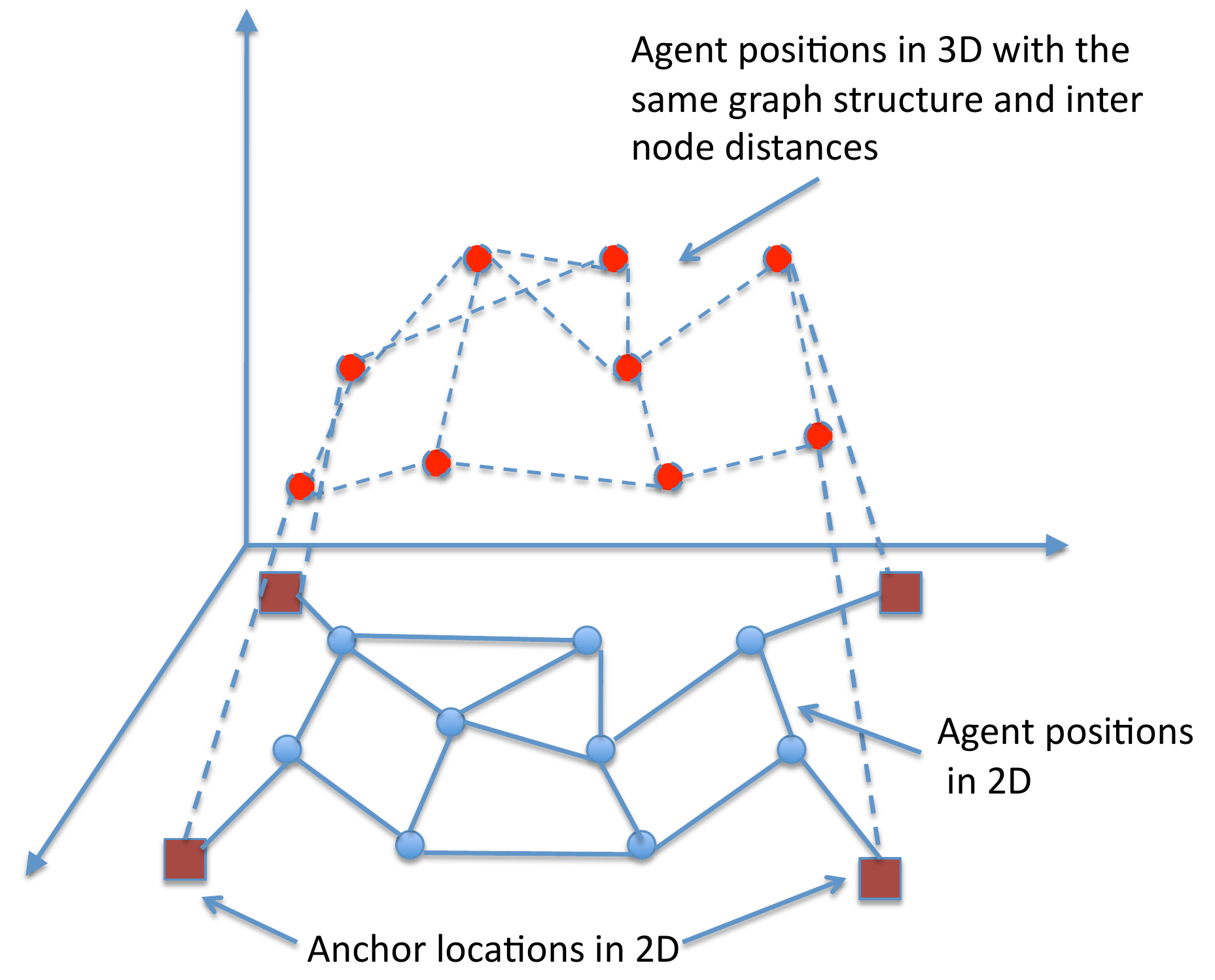}
 \caption{Illustration of the SDP relaxation. We search for points in a higher dimensional space satisfying the same distance measurements and the graph structure as in the lower dimensional space.}
 \label{fig:loc2d3d}
 \vspace*{-0.4in}
 \end{figure}
 
\subsection{Non-cooperative case}
 Consider the simple case of a single agent and multiple anchors, i.e. $N = 1$ and $M$ anchors. The minimization problem can be written as, 
 
   \begin{equation}
\begin{aligned}
& \underset{y, x, \{b_i\}}{\text{minimize}}
& & \sum_{i=1}^M \left(y-2a_i^Tx + a_i^Ta_i- \tilde{d}_{i}^2+b_{i}\right)^2 \\
& \text{subject to}
& & b_{i} \geq 0, \forall \ i\\
& & & y \geq x^Tx
\end{aligned}
\label{eqn:NonCoop}
\end{equation}
  where $x$ is the agent location to be estimated.
  
 The following theorem states our result,
 \begin{thm}
 For the case when the measurement noise is only due to bias errors i.e., $n_i = 0$, the minimization problem  (\ref{eqn:NonCoop}) gives an exact solution as long as the agent lies within the convex hull of two or more anchors with which it has LOS measurements.
\end{thm}

\begin{proof}
Let $x$ be the actual node location, and suppose it is contained within the convex hull of two or more anchor nodes with which it has LOS measurements. By contradiction, let $\hat{x}$ be a solution to equation (\ref{eqn:NonCoop}) such that $x \neq \hat{x}$. Let $W$ be the set of all $i$ such that $\tilde{d}_i^2$ is a LOS measurement and $x$ is contained in the convex hull of $\{a_i, i\in W\}$. Since the actual location of $x$ is a solution, the optimal objective value is upper (and lower) bounded by 0. Thus it must be true that $||\hat{x}-a_i||^2 \leq  (\tilde{d}_{i}^2 = d_i^2), \forall i \in W$. For this condition to hold, $\langle \hat{x}-x, a_i-\hat{x}\rangle > 0, \forall i \in W$ (consider 3 anchors with $x$ in the middle. To be a distinct solution, $\hat{x}$ must simultaneously be closer to all 3 anchors than $x_i$'s giving the inner product condition). But this is precisely the condition for the existence of a separating hyperplane between two convex sets, where the sets are the point $x$ and the convex hull of the LOS anchors. If there exists a separating hyperplane between these two sets, then $x$ is not in the convex hull of $\{a_i, i \in W\}$, which is a contradiction. So $x$ must be the unique solution.
\end{proof}

\subsection{Cooperative case}
\label{subsec:SDPcontrac}

In this section we consider the cooperative case where we have internode distances between the agents. Recall the optimization formulation in Algorithm 1. Note that for every term in the cost function we have a positive bias term $b_{ij}$, that we wish to estimate and use it to offset any NLOS bias in the corresponding distance measurement. The optimization formulation does not assume any prior knowledge on the bias magnitude or whether a particular link is LOS or NLOS.\\

 Consider a given network of anchors and agents ${\cal N}_1 = (X,A,{\cal E}_X,{\cal E}_A,D)$, with an underlying graph of connections between the nodes, where we have an edge between two nodes iff we are given a distance measurement between those two nodes. Suppose there exists another set of agent locations satisfying the same underlying graph topology as the original network ${\cal N}_1$ but with inter-node distances restricted to the edge set, shorter than the original network of nodes i.e., ${\cal N}_2 = (\tilde{X},A,{\cal E}_X,{\cal E}_A,\tilde{D})$, where $\tilde{D}_{ij} \leq D_{ij}, \forall (i,j) \in {\cal E}_X \cup {\cal E}_A$. Figure \ref{fig:ConNonCon}(a) is an example of such a network. In such a case, the optimization problem cannot distinguish whether the given set of distance measurements belonged to the original network, or whether the distance measurements were obtained from the second network corrupted by positively biased noise. Thus we can only hope to have a unique solution as long as the original network has a graph structure that rules out the existence of a different set of node locations satisfying the same graph structure and smaller inter-node distances. A fully connected graph is one such example. Figure \ref{fig:ConNonCon}(b) shows another example of such a network. Note that the node locations of the agents cannot be shifted without increasing at least one of the distances. Both these examples are in two-dimensions where we assumed that the nodes are also placed in two-dimensions (i.e. $k = 2$). However we saw that the SDP relaxation allows us to search in for node locations in a higher dimensional space. Hence we would need this property to hold in higher dimensions too. This motivates the definition of {\em non-contractible} networks.\\

\emph{Definition:} A network $(X,A,{\cal E}_X,{\cal E}_A,D)$ is said to be {\em non-contractible} if there exists a unique $X \in \mathbb{R}^{N \times k}$ satisfying the location constraints imposed by $D$ and there exists no $x'_j \in \mathbb{R}^h, \ \ j = 1,2,..,N$, $\forall h \geq k$ such that
\bean
||(a_{\ell};0) - x'_j||^2 & \leq & d_{\ell j}^2  \ \ \ \   \forall \ \ (\ell,j) \in {\cal E}_A\\
||x'_i - x'_j||^2 & \leq & d_{ij}^2  \ \ \ \  \forall \ \ (i,j) \in {\cal E}_X\\
x'_j & \neq & (x_j;0) \   \mbox{for some } j \in \{1, 2, 3, ....,N\}
\eean

\begin{figure}
\subfigure[Contractible network]{
\includegraphics[width=1.5in]{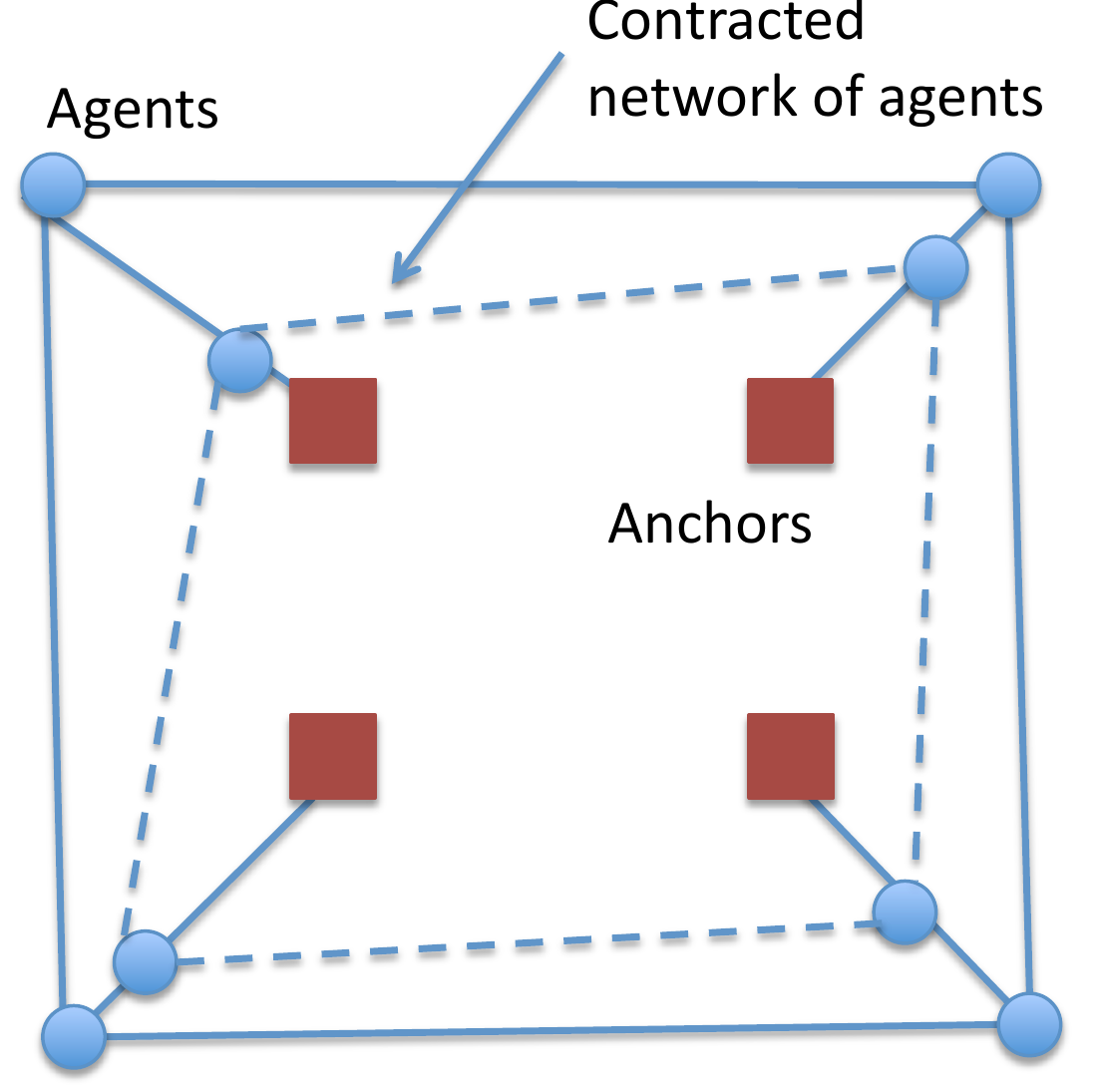}}
\subfigure[Non-contractible network]{
\includegraphics[width=1.55in]{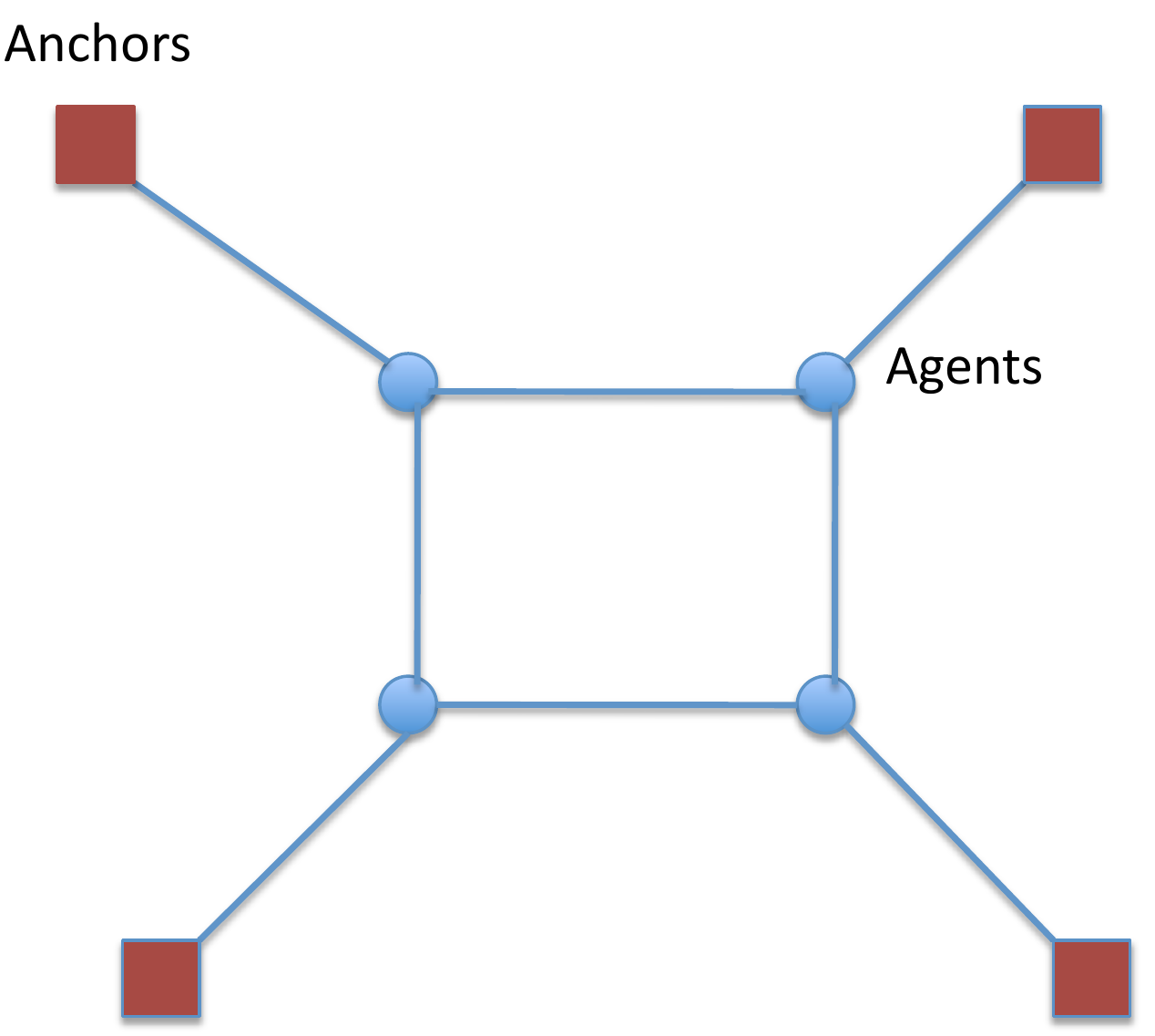}}
\caption{(a) Contractible networks. The dotted lines and corresponding nodes represent a contracted network with inter node distances smaller than that of the original network.(b) Non contractible network. See that the agent locations cannot be shifted without increasing at least one of the distances.
 }
\label{fig:ConNonCon}
\end{figure}

Figure \ref{fig:ConNonCon} shows examples of networks that are contractible and non-contractible in $\mathbb{R}^2$.  Note that even if one node lies outside the convex hull of the anchors, the network is contractible. However the converse does not hold. \\

The following theorem establishes that non-contractibility is a necessary and sufficient condition for the SDP relaxation to provide an exact and unique solution. 

\begin{thm}
\label{thm:SDPcontrac}
For the case when the measurement noise is only due to bias errors i.e. $n_{ij} = 0$,
\bit
\item Algorithm 1 gives the exact solution if the underlying LOS network of nodes is non-contractible.
\item If the max rank solution of the optimization problem in Algorithm 1 is $k$, then  the solution is exact.
\eit
\end{thm}
\begin{proof}
The intuition for the proof relies on the arguments made for the definition of non-contractibility. The sufficiency of the non-contractibility condition is proved using arguments of contradiction. See the appendix for a formal proof which is on the lines of \cite{biswas2004semidefinite}.\\
\end{proof}

{\em Remark:} The notion of non-contractibility is a generalization of the uniquely localizable condition \cite{biswas2004semidefinite} for NLOS localization. However, the question of what networks are non-contractible in practice is still an open question. There has been recent progress in the LOS case in characterizing the classes of networks that are uniquely localizable using graph rigidity theory \cite{zhu2010universal}. Analogously, one could hope to utilize results from tensegrity theory \cite{connelly2009tensegrities}, to characterize networks that are globally stable and non-contractible. This is a future direction of research.

\section{Validation}
\label{subsec:SDPvalidation}
\subsection{Simulation}
\begin{figure}
\centering
\includegraphics[width=2.5in]{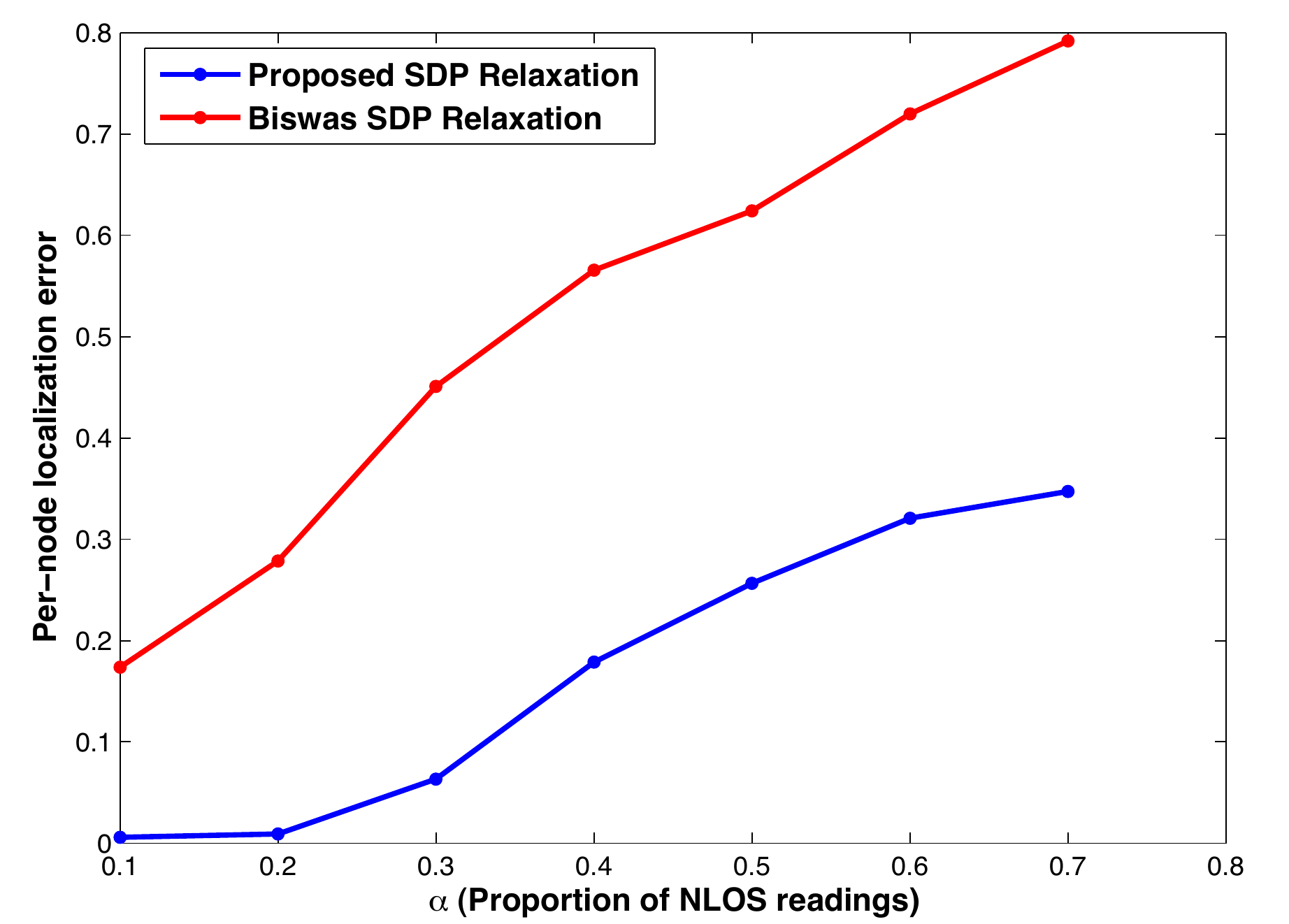}
\caption{Error as a function of the fraction of NLOS measurements under the different formulations. The NLOS bias was uniform on $[0,6]$, the noise was zero-mean Gaussian with $\sigma_{LOS} = 0.02$, and error values are averaged over 10 trials each.}
\label{fig:ErrorNLOS}
\end{figure}

\begin{figure}
\centering
\includegraphics[width=2.5in]{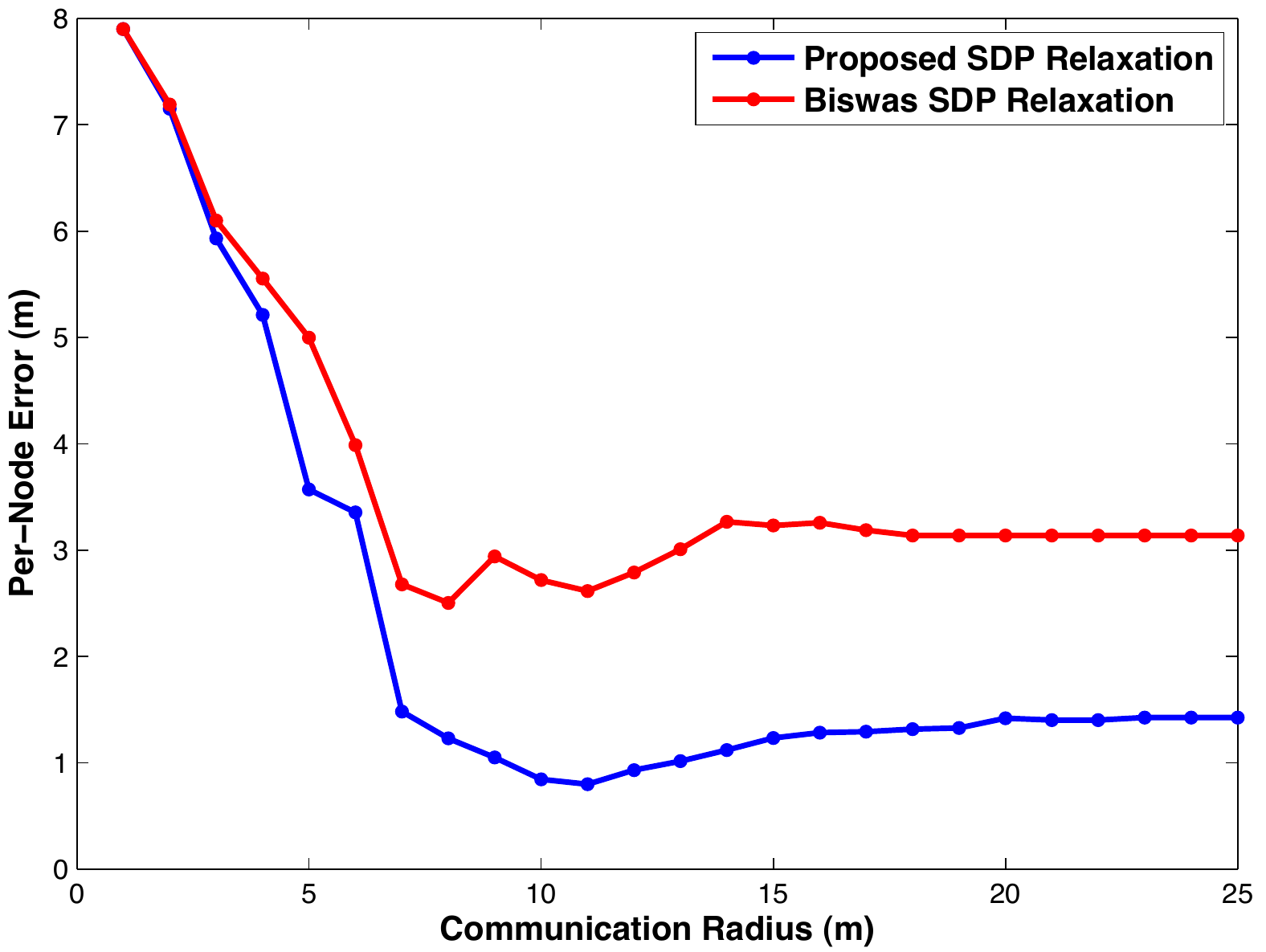}
\caption{Per-node error as a function of communication radius on the data set obtained by Patwari et al \cite{patwari2005locating}.}
\label{fig:patwari_error}
\end{figure}

Simulations were carried out for the proposed SDP relaxation and compared against the scheme proposed by Biswas et. al. \cite{biswas2004semidefinite}. We chose this baseline system to compare in order to highlight the gains obtained by explicitly accounting for NLOS in the measurements. The SDP relaxation in \cite{biswas2004semidefinite} is robust to small percentages of NLOS noise but fails for larger percentages given that they do not explicitly account for NLOS. The setup consisted of 50 nodes randomly placed in a $[-1,1]\times [-1,1]$ grid with 15 anchors. Anchors were placed uniformly along the upper and lower boundaries of the grid. The radius of connectivity was taken as $r = 1.5$, and the gaussian noise standard deviation was set to 0.02. The NLOS noise was taken to be uniform over the interval $[0,6]$. Results were obtained after averaging over 10 trials. \\

Figure \ref{fig:ErrorNLOS} shows the variation of the error in the node location as a function of the fraction of NLOS measurements for the relaxation. The error metric is the per node average mean squared error in the node location estimate i.e. $||x-\hat{x}||/\sqrt{N}$.  It can be seen that the proposed SDP relaxation is robust to a large fraction of NLOS errors.\\

\subsection{Experimental Results}
\label{subsubsec:SDPexpt}

\begin{figure}
\centering
\includegraphics[width=2.5in]{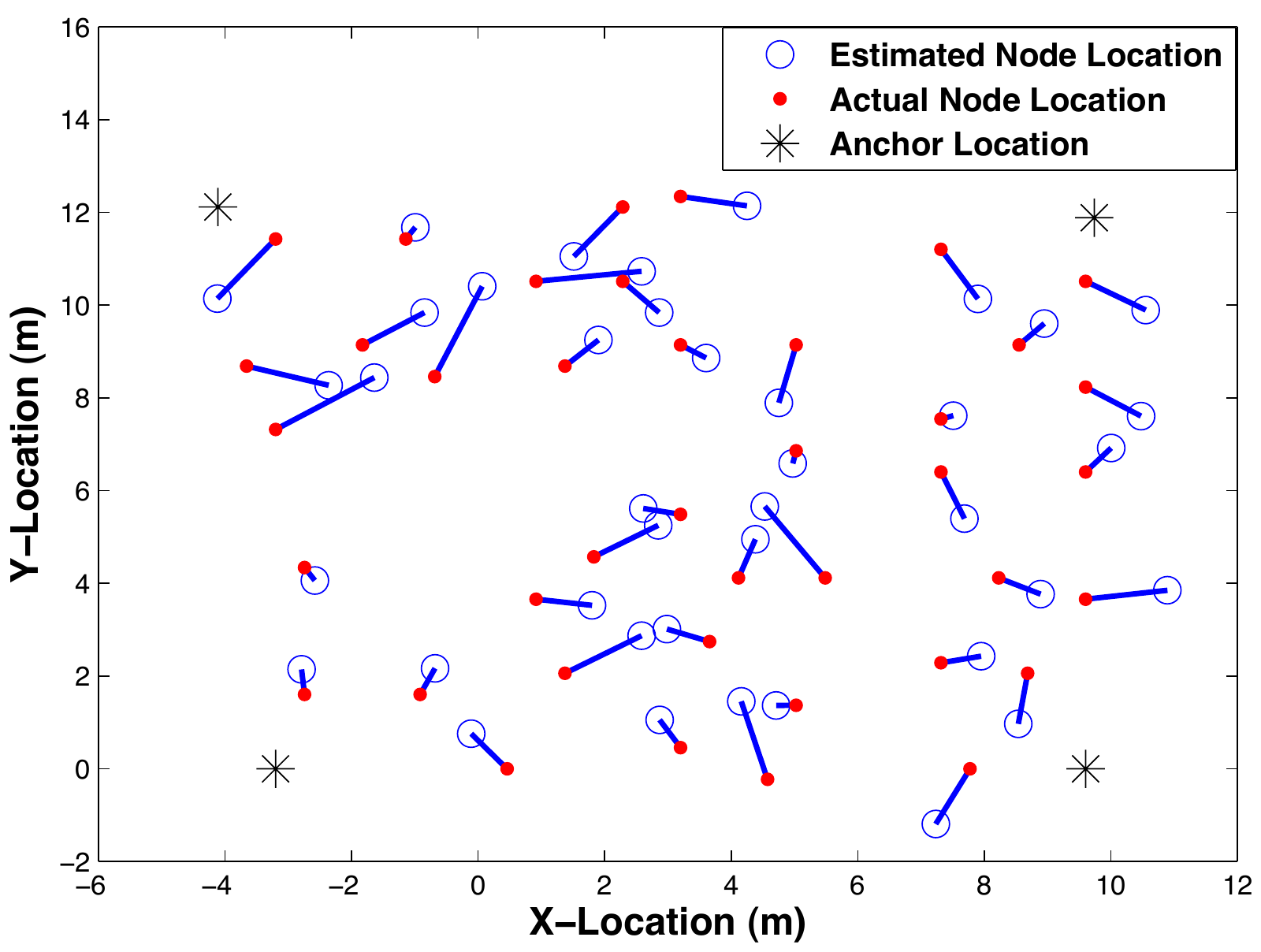}
\caption{Node location estimates for the experimental data set using our SDP relaxation.}
\label{fig: sdp_loc_patwari}
\end{figure}

\begin{figure}
\centering
\includegraphics[width=2.5in]{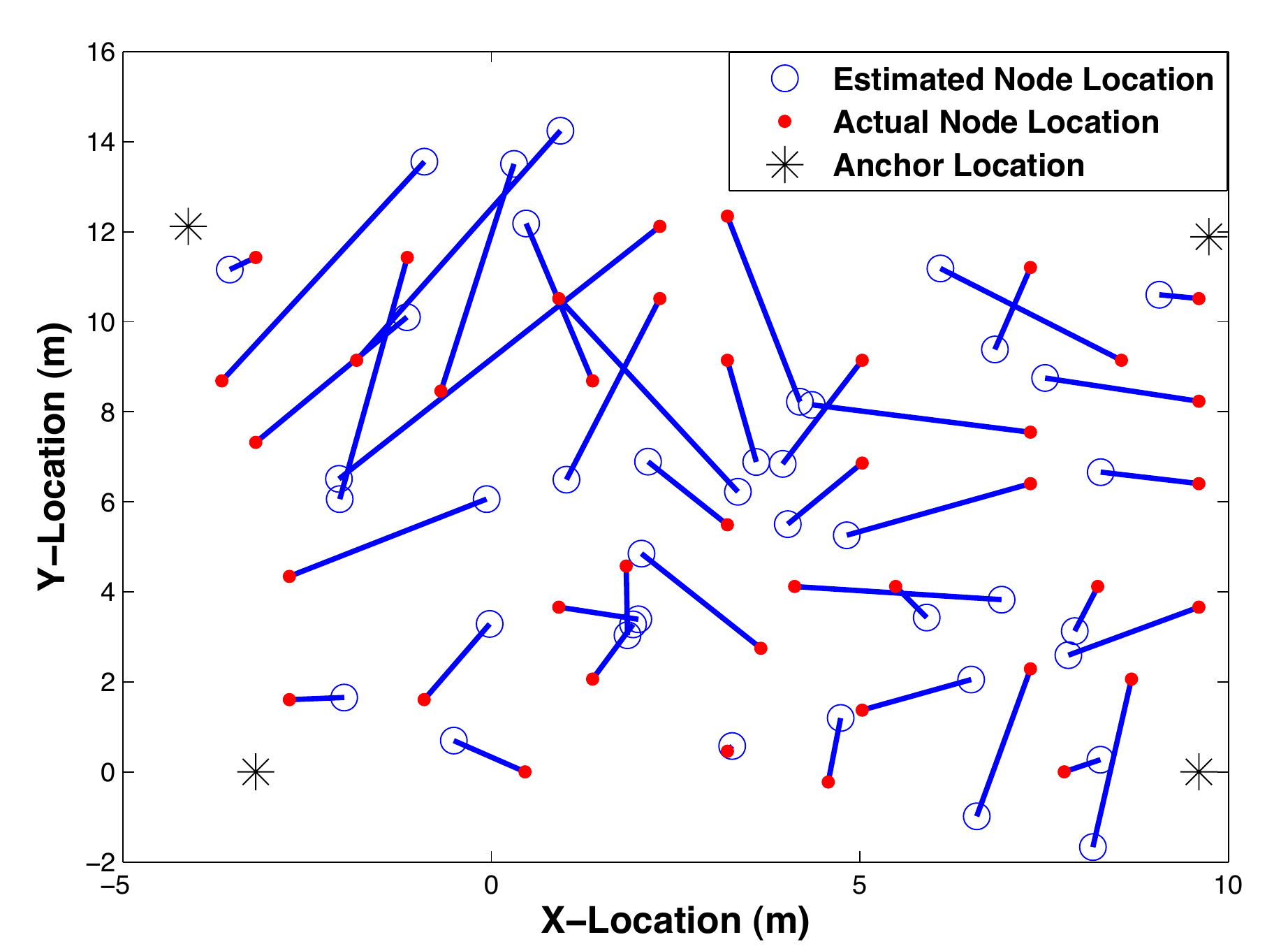}
\caption{Node location estimates for the experimental data set using the scheme by Biswas et al. \cite{biswas2004semidefinite}.}
\label{fig: biswas_loc_patwari}
\end{figure}

The SDP relaxation was tested against data obtained from real world experiments. We use the data set published by Patwari et. al. \cite{patwari2005locating}. The experiments were conducted in an indoor office environment by simulating 44 node locations using a transmitter and a receiver and obtaining pairwise time-of-arrival measurements. The environment had a lot of scatterers and nearly all the distance estimates have strong NLOS biases in them. The authors in \cite{patwari2005locating}, artificially subtracted out the NLOS biases and obtained the node location estimates using LOS algorithms. They report a per node localization error of $1.26m$ after utilizing all the measurements i.e. a fully connected network. We use the same data set without subtracting out the NLOS biases and show that the SDP relaxation performs very well even in this setup. The node location estimates are obtained using only a subset of measurements, based on an arbitrarily defined radius of communication. The SDP relaxation does not need any inputs regarding the noise parameters which is required by most of the existing algorithms in the literature. \\

Figure \ref{fig:patwari_error} shows the per node error as a function of the communication radius. The per node error is seen to be less than a meter even when the NLOS biases are included and only a subset of the measurements are used. However the performance declines as the radius becomes too large or too small. The reasoning is that for small radii, each node may not have enough information to fully resolve its position and the network could become non-contractible. For large radii, the multipath bias is more likely to be large, causing errors which can be reasoned as follows. The measurement model that we considered has the bias added to the square of the distances. However in practice the bias would be added to the actual distance value. Hence when we square the distances and apply our measurement model, the bias values would be a function of the distance and hence would increase as the distances increase. This partially explains the reason why the performance degrades slightly as we increase the radius of connectivity. Figure \ref{fig: sdp_loc_patwari} and Figure \ref{fig: biswas_loc_patwari} show the estimated and true node locations for the experimental data set obtained from the SDP relaxations.

\section{Conclusion}
In this paper, we considered the problem of estimating the locations of a set of points in a $k$-dimensional euclidean space given pairwise distance measurements  amongst the points. We considered the case when some fraction of the measurements could be arbitrarily corrupted, proposing a SDP relaxation formulation for the problem. The relaxation was shown to give the exact solution when the underlying graph is non-contractible. Simulations and testing on real-world data shows that the algorithm is quite robust to NLOS noise. Future work involves strengthening the theoretical results in  order to provide meaningful guarantees in practice. It would be interesting to characterize the classes of graphs that are non-contractible. We also need to theoretically prove that the SDP relaxation is robust to thermal noise.

\bibliographystyle{IEEEtran}
\bibliography{NLOSConvexRelax}

  \appendices
   \section{Proof of Theorem \ref{thm:SDPcontrac}}
   \label{app:SDPrelax}
\begin{lemma} 
If the max rank solution of  the optimization problem in Algorithm 1 is $k$, then it is unique.
\end{lemma}
\begin{proof}
We will prove this by contradiction. Suppose $(X_1,Y_1, B_1)$ and $(X_2,Y_2, B_2)$ are two solutions, then for any $\beta \leq 1$ one can easily check that $\beta(X_1,Y_1, B_1) + (1-\beta)(X_2,Y_2, B_2)$ is also a solution. Since the max rank is $k$, we should have that, 
\[\left[
\begin{array}{rl} \beta Y_1 + (1-\beta) Y_2 & \beta X_1^T + (1-\beta) X_2^T\\
                             \beta X_1 + (1-\beta) X_2  & I_k
 \end{array} \right]\] be rank $k$.  This gives us 
 \bean
  \beta Y_1 + (1-\beta) Y_2  & = & ( \beta X_1 + (1-\beta) X_2)^T (\beta X_1+ (1-\beta) X_2).
 \eean
 where $Y_1 = X_1^TX_1$ and $Y_2 = X_2^TX_2$. Rearranging terms, we get $||X_1-X_2||^2 = 0$ which implies $X_1 = X_2$.
 \end{proof} 
We will now show that if the underlying LOS network is non-contractible then the solution is unique. Assume that each node has at least one LOS measurement. We will prove uniqueness by contradiction. Suppose that the network is non-contractible and the solution is not exact. Then by the definition of non-contractibility and by the non-uniqueness of solution $X$, there exists a solution $X'$ of rank higher than $k$ that satisfies $Y \succeq X'^TX'$. $X'$ satisfies
 \[Y - X^TX = X'^{T}X'.\]
 Define
 \bean
  \tilde{X} & = & \left[
\begin{array}{rl} X\\
                            X'
 \end{array} \right].
 \eean
 We then have 
 \bean
 \tilde{X}^T\tilde{X} & = & X^TX + X'^TX'\\
 &  = & Y. 
 \eean
 Let ${\cal E}_{LOS}$ be the edge set of the LOS links. 
 \bean
 e_{ij}^T Y e_{ij} - \hat{d}_{ij}^2 + b_{ij} & = & 0, \\
  e_{ij}^T Y e_{ij} & \leq & d_{ij}^2, \ \ \ \forall i,j \in  {\cal E}_{LOS} \cap {\cal E}_X, \\
  ||\tilde{x}_i - \tilde{x}_j||^2 & \leq  & d_{ij}^2\ \ \  \forall i,j \in  {\cal E}_{LOS} \cap {\cal E}_X.
 \eean
 We also have, 
 \bean
 [e_i^T -a_{\ell}^T] \left[ {\begin{array}{cc}
 Y & X  \\
 X^T & I_k  \\
 \end{array} } \right] \left[
\begin{array}{rl} e_i\\
                            -a_{\ell}
 \end{array} \right] & \leq & d_{i\ell}^2 \ \ \  \\
 e_i^TYe_i - 2 a_{\ell}^T X e_i + a_{\ell}^Ta_{\ell} & \leq & d_{i\ell}^2,\\
   e_i^T \tilde{X}^T\tilde{X} e_i - 2 [a_{\ell}^T 0] 
 \left[
\begin{array}{rl} X\\
                            X'
 \end{array} \right] e_i + ||a_{\ell}||^2 & \leq & d_{i\ell}^2, \\
 ||(a_{\ell};0)- \tilde{x}_i||^2 & \leq& d_{i\ell}^2 \ \ \   \\
\mbox{for } \forall i,\ell & \in & {\cal E}_{LOS} \cap {\cal E}_A
 \eean
 Thus, we just showed that the underlying network is contractible which is a contradiction to our initial assumption that the network is non-contractible.

\end{document}

%% file: packages_all.tex
\usepackage{epsfig}
\usepackage{graphicx}
\usepackage{color}
\usepackage{chapterbib}
\usepackage{pdfsync}
\usepackage[hang]{subfigure}

\usepackage{times}
\usepackage{latexsym}
\usepackage{bbm,dsfont}
\usepackage{amsmath}
\usepackage{amssymb}
\usepackage{latexsym}
\usepackage{amsthm}
\usepackage{amsfonts}
\usepackage{epsfig}
\usepackage{latexsym}
\usepackage{cite}
\usepackage{graphicx}
\usepackage{amsmath}
\usepackage{amsthm}

\newcommand{\bea}{\begin{eqnarray}}
\newcommand{\eea}{\end{eqnarray}}

\usepackage{epsf}
\usepackage{amssymb}
\usepackage{graphicx}
\usepackage{amsmath}
\usepackage[english]{babel}
\usepackage{mathrsfs}

\newtheorem{thm}{Theorem}

\usepackage{setspace}

\input{newcommands}

\usepackage{ifthen}
\usepackage{epsfig,comment}
\usepackage{color}
\usepackage{amssymb}

\usepackage[newenum]{paralist}

\newcommand{\IsTR}{\ifthenelse{1<2}}

\IsTR{}{ \ninept }

%% file: newcommands.tex
\newcommand{\bfl}{\begin{flushleft}}
\newcommand{\efl}{\end{flushleft}}
\newcommand{\bfr}{\begin{flushright}}
\newcommand{\efr}{\end{flushright}}
\newcommand{\bc}{\begin{center}}
\newcommand{\ec}{\end{center}}

\newcommand{\edar}{\end{array}}
\newcommand{\begar}{\begin{array}}

\newcommand{\bit}{\begin{itemize}}
\newcommand{\eit}{\end{itemize}}
\newcommand{\beq}{\begin{equation}}
\newcommand{\eeq}{\end{equation}}
\newcommand{\ben}{\begin{enumerate}}
\newcommand{\een}{\end{enumerate}}
\newcommand{\bean}{\begin{eqnarray*}}
\newcommand{\eean}{\end{eqnarray*}}

\newtheorem{lemma}{Lemma}